\documentclass[11pt]{article}

\usepackage{fullpage,wrapfig}
\usepackage{amsmath, amsthm, amssymb, algorithm,thmtools,algpseudocode,enumerate,caption,framed}
\usepackage{paralist, sidecap}

\usepackage[usenames,dvipsnames,svgnames,table]{xcolor}
\definecolor{darkgreen}{rgb}{0.0,0,0.9}

\usepackage[colorlinks=true,pdfpagemode=UseNone,citecolor=Blue,linkcolor=Red,urlcolor=Black,pagebackref]{hyperref}

\usepackage{tikz}

\usepackage[T1]{fontenc}
\usepackage[utf8]{inputenc}
\usepackage{authblk}
\usepackage[colorinlistoftodos,prependcaption,textsize=tiny]{todonotes}

\newtheorem{theorem}{Theorem}[section]

\newtheorem{lemma}{Lemma}[section]
\newtheorem{observation}{Observation}[section]

\newtheorem{definition}{Definition}[section]

\newcommand{\bepg}{\textsc{B$_1$-EPG }}

\newcommand{\bvpg}{\textsc{B$_1$-VPG }}

\newcommand{\optds}{\textsc{OPT$_{DS}$}}
\newcommand{\opths}{\textsc{OPT$_{HS}$}}

\newcommand{\uchap}{$\ulcorner$}
\newcommand{\urast}{$\urcorner$}
\newcommand{\dchap}{$\llcorner$}
\newcommand{\drast}{$\lrcorner$}
\newcommand{\vtx}{\texttt{vtx}}
\newcommand{\edg}{\texttt{edg}}
\newcommand{\pivi}{$\mathcal{P}_{\vtx}$ }
\newcommand{\piei}{$\mathcal{P}_{\edg}$ }

\newcommand{\mds}{\textsc{Minimum Dominating Set }}
\newcommand{\mis}{\textsc{Maximum Independent Set }}

\newcommand{\apx}{\textsc{APX}-hard }
\newcommand{\mvc}{\textsc{Minimum Vertex Cover }}
\newcommand{\both}{\texttt{both}}
\newcommand{\one}{\texttt{one}}
\newcommand{\hpart}{\texttt{hPart}}
\newcommand{\vpart}{\texttt{vPart}}
\newcommand{\htip}{\texttt{hTip}}
\newcommand{\vtip}{\texttt{vTip}}
\newcommand{\hnei}{\texttt{hNeighbor}}
\newcommand{\vnei}{\texttt{vNeighbor}}
\newcommand{\cor}{\texttt{corner}}
\newcommand{\cro}{\texttt{cross}}

\title{Approximation Algorithms for Independence and Domination on \bvpg and \bepg Graphs}

\author{Saeed Mehrabi}
\affil{\small{University of Waterloo, Waterloo, Canada

				\url{smehrabi@uwaterloo.ca}}}

\date{}

\begin{document}

\maketitle

\begin{abstract}
A graph $G$ is called \emph{B$_k$-VPG} (resp., \emph{B$_k$-EPG}), for some constant $k\geq 0$, if it has a string representation on a grid such that each vertex is an orthogonal path with at most $k$ bends and two vertices are adjacent in $G$ if and only if the corresponding strings intersect (resp., the corresponding strings share at least one grid edge). If two adjacent strings of a B$_k$-VPG graph intersect exactly once, then the graph is called a \emph{one-string} B$_k$-VPG graph.

In this paper, we study the \mis and \mds problems on \bvpg and \bepg graphs. We first give a simple $O(\log n)$-approximation algorithm for the \mis problem on \bvpg graphs, improving the previous $O((\log n)^2)$-approximation algorithm of Lahiri et al.~\cite{LahiriMS15}. Then, we consider the \mds problem. We give an $O(1)$-approximation algorithm for this problem on one-string \bvpg graphs, providing the first constant-factor approximation algorithm for this problem. Moreover, we show that the \mds problem is \apx on \bepg graphs, ruling out the possibility of a PTAS unless \textsc{P=NP}. Finally, we give constant-factor approximation algorithms for this problem on two non-trivial subclasses of \bepg graphs. To our knowledge, these are the first results for the \mds problem on \bepg graphs, partially answering a question posed by Epstein et al.~\cite{EpsteinGM13}.
\end{abstract}

\section{Introduction}
\label{sec:introduction}
In this paper, we study \mis and \mds on \bvpg and \bepg graphs. These are two special subclasses of \emph{string graphs}, which are of interest in several applications such as circuit layout design and bioinformatics. In a string representation of a graph, the vertices are drawn as a set of curves in the plane and two vertices are adjacent if their corresponding curves intersect; graphs that can be represented in this way are called string graphs. Studying string graphs dates back to 1970s when Ehrlich et al.~\cite{EhrlichET76} showed that every planar graph has a string representation. Also, Scheinerman's conjecture~\cite{thesisEdwardScheinerman}, stating that all planar graphs can be represented as intersection graphs of line segments, was a long-standing open problem until 2009 when it was proved affirmatively by Chalopin and Goncalves~\cite{ChalopinG09}.

EPG graphs (stands for Edge intersection of Paths in a Grid) were introduced by Golumbic et al.~\cite{GolumbicLS09} as the class of graphs whose vertices can be represented as simple orthogonal paths on a rectangular grid in such a way that two vertices are adjacent if and only if the corresponding paths share at least one edge of the grid. Golumbic et al.~\cite{GolumbicLS09} showed that every graph is an EPG graph, and the size of the underlying grid is polynomial in the size of $G$. Since then much of the work has focused on studying subclasses of EPG graphs; in particular, restricting the type of paths allowed. A turn of a path at a grid node is called a \emph{bend} and a graph is called a \emph{B$_k$-EPG graph} if it has an EPG representation in which each path has at most $k$ bends. In this paper, we are interested in B$_k$-EPG graphs for $k=1$.
\begin{definition}[\bepg Graph]
\label{def:B1EPG}
A graph $G=(V,E)$ is called a \emph{\bepg graph}, if every vertex $u$ of $G$ can be represented as a \emph{path} $P_u$ on a grid $\mathcal{G}$ such that \begin{inparaenum}[(i)] \item $P_u$ is orthogonal and has at most one bend, and \item paths $P_u, P_v\in P$ share a grid edge of $\mathcal{G}$ if and only if $(u, v)\in E$. \end{inparaenum}
\end{definition}

Instead of considering the edge intersection of graphs on a grid, we can think of the vertex intersection of graphs on a grid. More formally, a graph is said to have a VPG representation (stands for Vertex representation of Paths in a Grid), if its vertices can be represented as simple orthogonal paths on a rectangular grid such that two vertices are adjacent if and only if the corresponding paths share at least one grid node. Although these graphs were considered a while ago when studying string graphs~\cite{MiddendorfP92}, they were formally investigated by Asinowski et al.~\cite{AsinowskiCGLLS12}. Similar to B$_k$-EPG graphs, a \emph{B$_k$-VPG graph} is a VPG graph in which each path has at most $k$ bends. In this paper, we are interested in B$_k$-VPG graphs for $k=1$.
\begin{definition}[\bvpg Graph]
\label{def:B1VPG}
A graph $G=(V, E)$ is called a \emph{\bvpg graph}, if every vertex $u$ of $G$ can be represented as a \emph{path} $P_u$ on a grid $\mathcal{G}$ such that \begin{inparaenum}[(i)] \item $P_u$ is orthogonal and has at most one bend, and \item two paths $P_u$ and $P_v$ intersect each other at a grid node if and only if $(u, v)\in E$. \end{inparaenum}
\end{definition}

We remark that by \emph{intersecting} each other, we exclude the case where two paths only \emph{touch} each other. Figure~\ref{fig:sampleGraphs} shows a graph with its \bvpg and \bepg representations. A natural restriction on \bvpg graphs can be the number of intersections allowed between every pair of paths. A string graph is called \emph{one-string} if it has a string representation in which curves intersect at most once~\cite{ChalopinGO10,BiedlD15}. By combining one-string and \bvpg representations, a \emph{one-string \bvpg graph} is defined as a \bvpg graph in which two paths intersect each other exactly once whenever the corresponding vertices are adjacent.
\begin{figure}[t]
\centering
\includegraphics[width=0.90\textwidth]{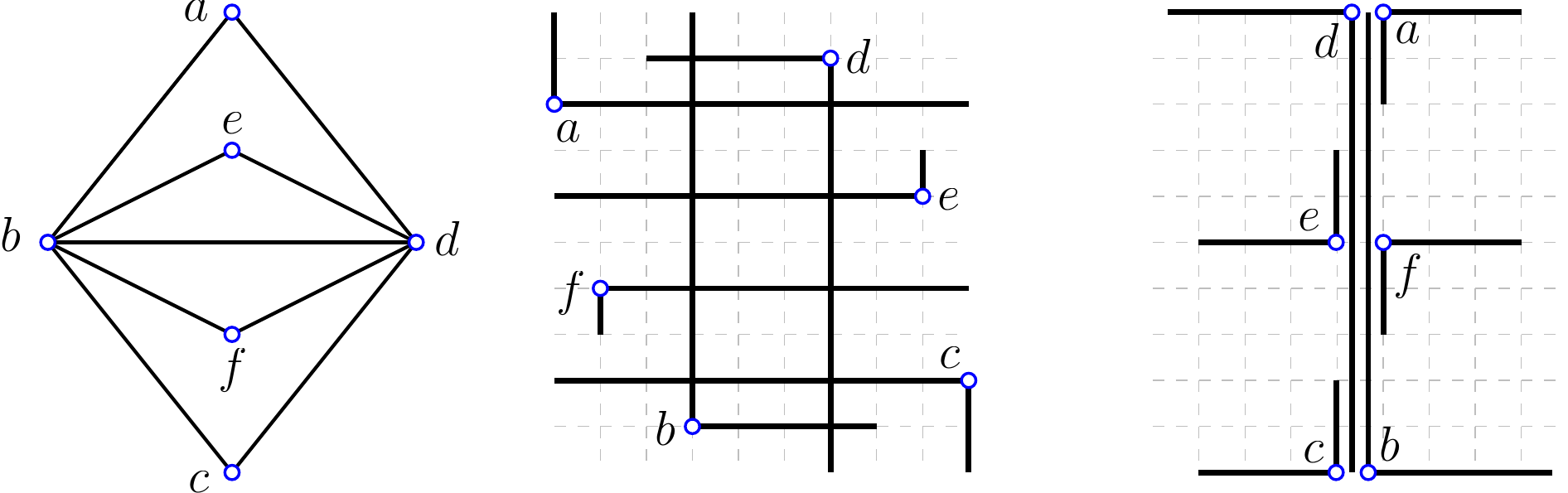}
\caption{A graph on six vertices (left) with its \bvpg (middle) and \bepg (right) representations. Notice that the vertices $e$ and $f$ are not adjacent in the \bepg representation as they only share a grid node (but not a grid edge).}
\label{fig:sampleGraphs}
\end{figure}

In this paper, we study two classical graph problems; namely, the \mis and \mds problems, on \bepg and \bvpg graphs. Let $G=(V, E)$ be an unweighted, undirected graph. A set $S\subseteq V$ is an independent set if no two vertices in $S$ are adjacent, while $S$ is a dominating set if every vertex in $V\setminus S$ is adjacent to some vertex in $S$. The objective of the \mis problem is to compute a maximum-cardinality independent set of $G$, whereas the goal in \mds is to minimize the size of the dominating set $S$. \mis and \mds are two fundamental optimization problems in graph theory, which arise in many applications such as wireless sensor networks, scheduling and resource allocation; both problems are well known to be \textsc{NP}-hard.

It is known that every graph has an EPG representation~\cite{GolumbicLS09}. Heldt et al.~\cite{HeldtKU14} showed that every graph of maximum degree $\Delta$ has a B$_\Delta$-EPG representation. Moreover, Heldt et al.~\cite{HeldtKU14a} proved that the class of B$_4$-EPG graphs contains all planar graphs; it is however still open whether $k=4$ is best possible. They also proved a conjecture of Biedl and Stern~\cite{BiedlS10} affirmatively stating that all outerplanar graphs are B$_2$-EPG graphs. For the case of \bepg graphs, Golumbic et al.~\cite{GolumbicLS09} showed that every tree is a \bepg graph, and a characterization of the subfamily of co-graphs that are \bepg is also known by a complete family of minimal forbidden induced sub-graphs~\cite{CohenGR14}. We observe by the construction of Ehrlich et al.~\cite{EhrlichET76} that every planar graph has a VPG representation. Moreover, Chaplick and Ueckerdt~\cite{ChaplickU13} proved that every planar graph has a B$_2$-VPG representation, although the representation may not be one-string. On the other hand, Chalopin et al.~\cite{ChalopinGO10} showed that every planar graph has a one-string representation. Biedl and Derka~\cite{BiedlD15} re-proved these results simultaneously by showing that every planar graph has a one-string B$_2$-VPG representation. It is however not known if every planar graph has also a one-string B$_1$-VPG representation. Felsner et al.~\cite{FelsnerKMU16} proved that every planar 3-tree is \bvpg and the line graph of every planar graph is a \bvpg graph.

In terms of recognition problems, it is known that recognizing \bepg graphs is \textsc{NP}-hard~\cite{HeldtKU14}, and the hardness of recognizing B$_2$-EPG graphs is shown by Pergel and Rz{\k{a}}{\.{z}}ewski~\cite{PergelR16}. Moreover, recognizing VPG graphs is \textsc{NP}-complete~\cite{ChaplickJKV12}. They also showed that recognizing if a given B$_{k+1}$-VPG graph is a B$_k$-VPG graph is \textsc{NP}-complete even if we are given a B$_{k+1}$-VPG representation as part of the input. Furthermore, Chaplick et al.~\cite{ChaplickCS11} considered the recognition problem for some restricted classes of B$_0$-VPG graphs.

\paragraph{Related work and our results.} The decision version of the \mis problem on \bvpg graphs was shown to be \textsc{NP}-complete by Lahiri et al.~\cite{LahiriMS15} who also gave an $O((\log n)^2)$-approximation algorithm for this problem. Notice that the best known algorithm for \mis on arbitrary string graphs has an approximation factor $n^{\epsilon}$, for some $\epsilon>0$~\cite{FoxP11}.
\begin{itemize}
\item We give a polynomial-time $O(\log n)$-approximation algorithm for the \mis problem on \bvpg graphs, improving the $O((\log n)^2)$-approximation algorithm of Lahiri et al.~\cite{LahiriMS15}.
\end{itemize}

Our algorithm is actually based on that of Lahiri et al.~\cite{LahiriMS15}, which is a divide-and-conquer approach~\cite{clrsBook,deBergBook}. However, instead of solving a special subproblem using a secondary divide-and-conquer, we show how to solve that subproblem directly and improve the approximation factor.
For \mds on \bvpg graphs, Asinowski et al.~\cite{AsinowskiCGLLS12} proved that every circle graph is a \bvpg graph. A circle graph is the intersection graph of a set of chords of a circle; that is, each vertex corresponds to a chord of the circle and two vertices are adjacent if and only if their corresponding chords intersect. The proof of Asinowski et al.~\cite{AsinowskiCGLLS12} shows that every circle graph is in fact a one-string \bvpg graph (although not every one-string \bvpg graph is a circle graph~\cite{Bouchet94}). Since \mds is \apx on circle graphs~\cite{DamianP06}, the problem becomes \apx also on one-string \bvpg graphs.\footnote{We note here that, unlike the \mds problem, \mis is polynomial-time solvable on circle graphs~\cite{NashG10}.} However, to our knowledge, there is no approximation algorithm known for the problem. We note that there are $O(1)$-approximation algorithms for \mds on circle graphs~\cite{Damian-IordacheP99a,Damian-IordacheP02}, but these algorithms do not work for \bvpg graphs as they heavily rely on the fact that the vertices of the input graph are modelled as chords of a circle.
\begin{itemize}
\item We give a polynomial-time $O(1)$-approximation algorithm for the \mds problem on \emph{one-string} \bvpg graphs. To our knowledge, this is the first constant-factor approximation algorithm for this problem on \bvpg graphs.
\end{itemize}

Our algorithm is based on formulating the problem as a hitting set problem among orthogonal line segments and then proving that the corresponding hitting set has a small $\varepsilon$-net. However, proving the existence of such an $\varepsilon$-net is not straightforward and requires a decomposition of the problem into two ``one-dimensional'' instances, each of which will then admit such an $\varepsilon$-net. Informally speaking, the idea is that although the \mds problem is not ``decomposable'' into horizontal and vertical instances, the $\varepsilon$-net corresponding to this problem is. Combining these nets along with the technique of Br\"{o}nnimann and Goodrich~\cite{BronnimannG95} gives us the desired result.

For \bepg graphs, the \mis problem is known to be \textsc{APX}-hard~\cite{BougeretBGP15}, and there exists a polynomial-time 4-approximation algorithm for this problem~\cite{EpsteinGM13}. Moreover, there exists a 4-approximation algorithm for the \mds problem on \bepg graphs~\cite{ButmanHLR10}, as it is known that such graphs are a subset of 2-interval graphs~\cite{HeldtKU14}. However, we were unable to find a reference on the complexity of the \mds problem on \bepg graphs. We note here that every tree is a \bepg graph~\cite{GolumbicLS09}, and even permutation graphs or co-graphs are known to be \bepg\cite{CohenGR14}, but \mds is polynomial-time solvable on these graphs~\cite{graphTheoryBook}.
\begin{itemize}
\item We prove that the \mds problem is \apx on \bepg graphs, even if only two types of paths are allowed in the input graph. Thus, there exists no PTAS for this problem on \bepg graphs unless \textsc{P=NP}.
\end{itemize}

Let us also mention that Asinowski and Ries~\cite{AsinowskiR12} showed that the class of \bepg graphs contains a number of subclasses of chordal graphs; namely, chordal bull-free, chordal claw-free and chordal diamond-free graphs; however, we were unable to find any reference on the status of the \mds problem on these subclasses of chordal graphs. To show the \textsc{APX}-hardness, we give an \textsc{L}-reduction from the \mvc problem on graphs with maximum-degree three (which is known to be \apx~\cite{AlimontiK00}) to the \mds problem on \bepg graphs. As noted above, there is a 4-approximation algorithm for the \mds problem on \bepg graphs~\cite{ButmanHLR10}. Here, we give $c$-approximation algorithms for the \mds problem on two non-trivial subclasses of \bepg graphs, for $c\in\{2, 3\}$. Let us defer the definition of these subclasses to Section~\ref{subsec:MDSonBEPG}.
\begin{itemize}
\item We give polynomial-time constant-factor approximation algorithms for the \mds problem on two subclasses of \bepg graphs.
\end{itemize}

\paragraph{Organization.} This paper is organized based on the problems as follows. We first give some notations and definitions in Section~\ref{sec:prelimins}. Then, we give our $O(\log n)$-approximation algorithm for the \mis problem on \bvpg graphs in Section~\ref{sec:independentSet}. Next, we consider the \mds problem on \bvpg and \bepg graphs in Section~\ref{sec:dominatingSet}. We first give the $O(1)$-approximation algorithm for this problem on one-string \bvpg graphs in Section~\ref{subsec:MDSonBVPG}, and then present our results for this problem on \bepg graphs in Section~\ref{subsec:MDSonBEPG}. We conclude the paper with a discussion on open problems in Section~\ref{sec:conclusion}.

\section{Notation and Definitions}
\label{sec:prelimins}
For a \bvpg (resp., \bepg) graph $G$, we use $\langle$\pivi$,\mathcal{G}\rangle$ (resp., $\langle$\piei$,\mathcal{G}\rangle$) to denote a particular \bvpg (resp., \bepg) representation of $G$, where \pivi (resp., \piei) is the collection of paths correspond to the vertices of $G$ and $\mathcal{G}$ is the underlying grid. Since the recognition problem is \textsc{NP}-hard on such graphs, we assume throughout this paper, that we are always given a string representation of a \bvpg or \bepg graph as input (in addition to $G$). We sometimes violate the wording and say \emph{path(s) in $G$} to actually refer to the vertices in $G$ corresponding to the paths in \pivi or \piei. We denote the $x$- and $y$-coordinates of a point $p$ in the plane by $x(p)$ and $y(p)$, respectively.

Let $G=(V, E)$ be a graph with its representations $\langle$\pivi$,\mathcal{G}\rangle$ and $\langle$\piei$,\mathcal{G}\rangle$, and let $P\in\{$\pivi$\cup$\piei$\}$ be an arbitrary path. Since $P$ has at most one bend, it is in one of the types \{\uchap, \dchap, \urast, \drast, $|$, $-$\}. We call a path $P$ of type $\texttt{x}\in$\{\uchap, \dchap, \urast, \drast\}, a \emph{$\texttt{x}$-type path}; we complete the definition by referring to no-bend paths as \dchap-type paths.

We denote the horizontal and vertical segments of $P$ by $\hpart(P)$ and $\vpart(P)$, respectively. We call the common endpoint of $\hpart(P)$ and $\vpart(P)$ the \emph{corner} of $P$ and denote it by $\cor(P)$. Moreover, let $\htip(P)$ (resp., $\vtip(P)$) denote the endpoint of $\hpart(P)$ (resp., $\vpart(P)$) that is not shared with $\vpart(P)$ (resp., not shared with $\hpart(P)$). Let $N[P]$ denote the set of paths adjacent to $P$; we assume that $P\in N[P]$. Moreover, for a set $S$ of paths, define $N[S]:=\cup_{P\in S}N[P]$. We denote the set of neighbours of $P$ that share at least one grid edge with $\hpart(P)$ (resp., with $\vpart(P)$) by $\hnei(P)$ (resp., by $\vnei(P)$). Throughout this paper, we assume a weak general-position assumption on the corners of paths in \piei only: for every two distinct paths $P, P'\in$\piei, either $x(\cor(P))\ne x(\cor(P'))$ or $y(\cor(P))\ne y(\cor(P'))$; that is, the corners of no two paths in \piei coincide. Notice that by our weak general-position assumption, $\{\hnei(P), \vnei(P)\}$ is a partition of $N[P]$ for any $P\in$\piei.

\section{Independent Set}
\label{sec:independentSet}
In this section, we consider the \mis problem on \bvpg graphs, and give an $O(n^3)$-time $O(\log n)$-approximation algorithm for this problem. This improves the previous $O((\log n)^2)$-approximation algorithm of Lahiri et al.~\cite{LahiriMS15}. Our algorithm is actually based on that of Lahiri et al.~\cite{LahiriMS15} and uses a divide-and-conquer approach. We first describe their algorithm and then show how a simple modification can improve the approximation factor, although the running time would slightly increase. For the rest of this section, let $G$ be a \bvpg graph with $n$ vertices, and let $\langle$\pivi$,\mathcal{G}\rangle$ denote a \bvpg representation of it.

\paragraph{Approximating a subproblem.} To solve \mis on \bvpg graphs, Lahiri et al.~\cite{LahiriMS15} first considered a subproblem in which all paths in \pivi are assumed to be of \dchap-type. So, suppose for now that $G$ is a \bepg graph such that all paths in \pivi are of \dchap-type. The algorithm sorts the paths in \pivi by the $x(\cor(P))$ of all paths $P$ and considers the vertical line $\ell: x=x_{med}$, where $x_{med}$ is the middle point between $x(\cor(P_{\lfloor n/2\rfloor}))$ and $x(\cor(P_{\lfloor n/2\rfloor+1}))$ in the ordering. Then, the paths are partitioned into three groups: the set $I_L$ of paths that lie to the left of $\ell$, the set $I_R$ of paths that lie to the right of $\ell$ and the set $I_M$ of paths that intersect $\ell$. The algorithm solves \mis on $I_L$ and $I_R$ recursively and returns the larger of $S(I_L)\cup S(I_R)$ and $S(I_M)$ as the solution, where $S(A)$ is a maximum independent set on the set of paths in $A$.

To compute $S(I_M)$, the algorithm performs the above divide-and-conquer approach on the sorted $y(\cor(P))$ of only paths $P\in I_M$; hence, partitioning the paths in $I_M$ into three analogous groups $I'_L$, $I'_R$ and $I'_M$. Notice that the paths in $I'_M$ are intersecting a vertical and a horizontal line. Lahiri et al.~\cite{LahiriMS15} showed that the \bvpg graph induced by the paths in $I_M\cap I'_M$ is a co-comparability graph for which \mis can be solved optimally in polynomial time~\cite{KohlerM16,Golumbic04}. They analyzed this algorithm and showed that it has an $O((\log n)^2)$-approximation factor.

To improve the approximation factor to $O(\log n)$, we avoid running the divide-and-conquer algorithm on the $y(\cor(P))$ of paths $P\in I_M$. Instead, we directly solve \mis on paths in $I_M$ and then return the larger of $S(I_L)\cup S(I_R)$ and $S(I_M)$ as the solution. To solve \mis on $I_M$, we show that the \bvpg graph $G_{I_M}$ induced by the paths in $I_M$ is an \emph{outerstring graph} for which we can solve \mis optimally using the recent algorithm of Keil et al.~\cite{KeilMPV15}. See Algorithm~\ref{alg:approxMISonVPG} in which \textsc{MIS}$(H)$ denotes the algorithm of Keil et al.~\cite{KeilMPV15} for computing a maximum independent set on an outerstring graph $H$.

It remains to show that $G_{I_M}$ is an outerstring graph. Outerstring graphs were first introduced by Kratochv\'{i}l~\cite{Kratochvil91,Kratochvil91a} as the intersection graphs of a set of curves in the plane that lie inside a disk such that each curve intersects the boundary of the disk in one of its endpoints. Keil et al.~\cite{KeilMPV15} gave an $O(N^3)$-time algorithm for \mis on outerstring graphs, where they model the disk as a polygon and each curve as a polygonal line (i.e., a chain of line segments) completely inside the polygon such that one of its endpoints coincides with a vertex of the polygon. Here, $N$ denotes the total number of segments used in the geometric representation of the polygonal lines and the polygon. The main part of the proof of Lemma~\ref{lem:misOnIM} is to show that $G_{I_M}$ is an outerstring graph.
\begin{lemma}
\label{lem:misOnIM}
\mis can be solved on $G_{I_M}$ in $O(N^3)$ time, where $N$ is the number of paths in $I_M$.
\end{lemma}
\begin{proof}
We prove that $G_{I_M}$ is an outerstring graph with an $O(N)$ geometric representation. The lemma then follows by the algorithm of Keil et al.~\cite{KeilMPV15} for solving \mis on outerstring graphs.

Enclose the paths in $I_M$ within a large-enough vertical rectangle $R$ such that the vertical line $\ell$ contains the right side of $R$; ignore the portion of every path of $I_M$ that lies to the right of $\ell$ (and hence outside of $R$). Next, convert $R$ into a polygon $P$ by slightly modifying $R$ at the intersection point of each path in $I_M$ with the right side of $R$ such that the path ends at a vertex of the polygon $P$; see Figure~\ref{fig:ourOuterstring} for an illustration. First, notice that by definition of $\ell$, every path in $I_M\cap R$ intersects the right side of $R$ in exactly one point. Moreover, since all paths in $I_M$ are \dchap-type paths, no two paths in $I_M$ would intersect each other on the right of $\ell$ without first intersecting each other inside $R$; hence, $R$ captures all the adjacencies between the paths in $I_M$. The resulting representation models our outerstring graph.

Clearly, $P$ is a simple polygon with $O(N)$ vertices (and, hence $O(N)$ edges). Moreover, each path inside $P$ consists of two segments and so the total number of segments to represent all the paths in $I_M$ is $O(N)$. Therefore, $G_{I_M}$ is outerstring and the above construction is an $O(N)$ geometric representation of $G_{I_M}$.
\end{proof}

\begin{figure}[t]
\centering
\includegraphics[page=1,width=0.80\textwidth]{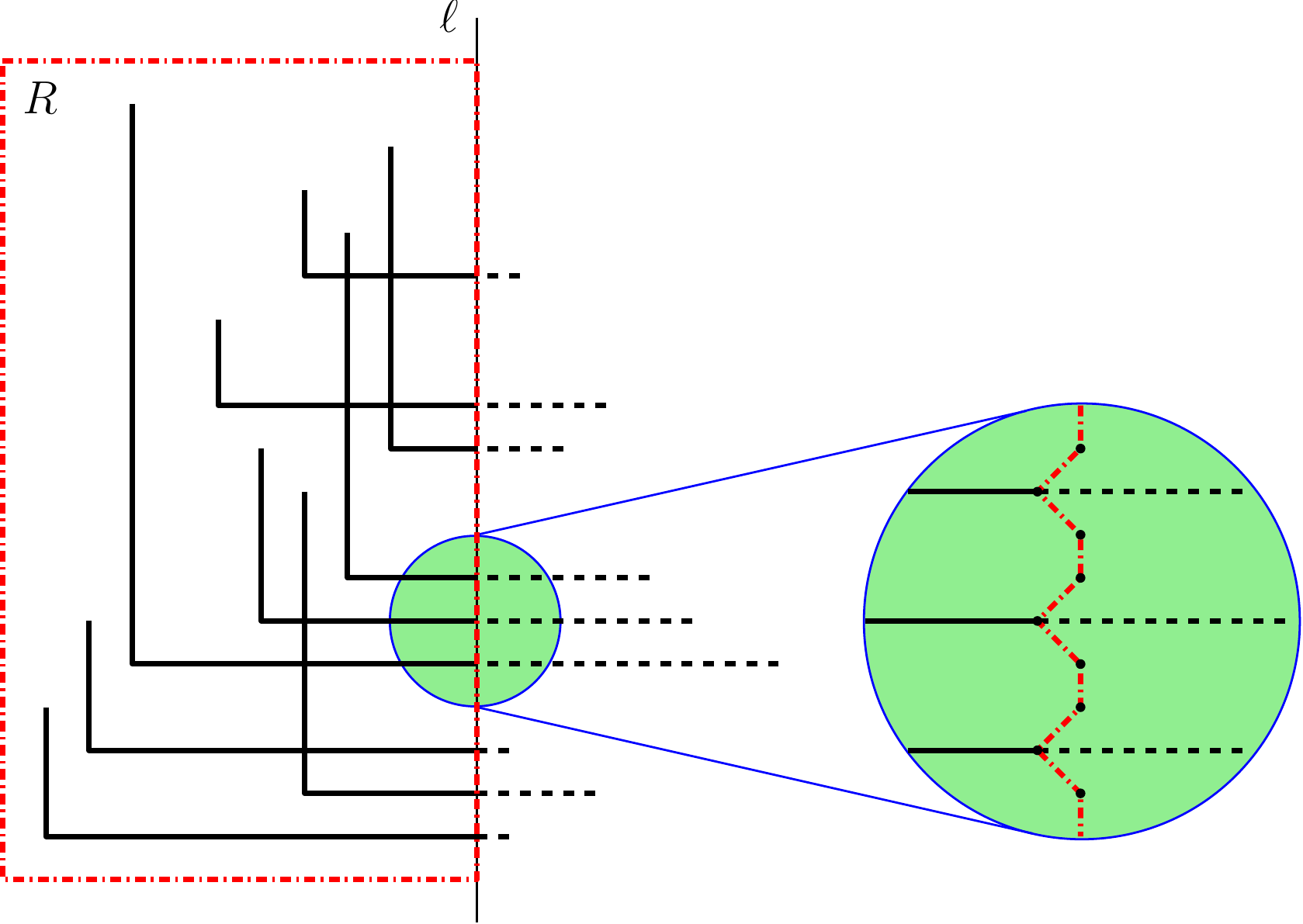}
\caption{A illustration in support of the proof of Lemma~\ref{lem:misOnIM}.}
\label{fig:ourOuterstring}
\end{figure}

\begin{algorithm}[t]
\caption{\textsc{ApproximateMISonVPG}($G$)}
\label{alg:approxMISonVPG}
\begin{algorithmic}[1]
\If{$n\leq 2$}
	\State compute a maximum independent set $S$ of $G$ in $O(1)$ time;
\Else
	\State compute $x_{med}$ and the partition $\{I_L, I_R, I_M\}$;
	\State $S(I_L)\gets$ ApproximateMISonVPG($I_L$);
	\State $S(I_R)\gets$ ApproximateMISonVPG($I_R$);
	\State $S(I_M)\gets$\textsc{MIS}$(I_M)$;
	\State $S\gets$ the larger of $S(I_L)\cup S(I_R)$ and $S(I_M)$;
\EndIf
\State\Return $S$;
\end{algorithmic}
\end{algorithm}

\begin{lemma}
\label{lem:approxMISonVPGfactor}
Algorithm~\ref{alg:approxMISonVPG} is an $O(n^3)$-time $O(\log n)$-approximation algorithm for the \mis problem on $G$.
\end{lemma}
\begin{proof}
Clearly, the paths in $I_L$ are disjoint from those in $I_R$; so, $S(I_L)\cup S(I_R)$ is an independent set. Therefore, the solution returned by the algorithm is a feasible solution for \mis on $G$. Let $T(n)$ denote the running time of the algorithm on any \bvpg graph $G$ with $n$ vertices. If $n\leq 2$, then clearly $T(n)=O(1)$. Otherwise, we first compute $x_{med}$ in $O(n\log n)$ time by sorting all the paths $P$ by $x(\cor(P))$. Then, by the choice of $x_{med}$, we solve two instances of the problem recursively each of which having size $O(n/2)$. Since the algorithm of Keil et al.~\cite{KeilMPV15} takes $O(n^3)$ time, we have $T(n)=2T(n/2)+O(n^3)$, which is solved to $T(n)=O(n^3)$.

Now, let $O^*$ be a maximum independent set in $G$. Moreover, let $O^*_L$, $O^*_R$ and $O^*_M$ denote a maximum independent set for the paths in $I_L$, $I_R$ and $I_{M}$, respectively. We prove by induction on $n$ that $\lvert S\rvert\geq\lvert O^*\rvert/\log n$, hence completing the proof of the lemma. If $n\leq 2$, then the lemma is clearly true. Suppose that the lemma is true for all $m<n$. Notice that since the algorithm directly computes a maximum independent set of $I_M$, we have $\lvert S(I_M)\rvert=\lvert O^*_M\rvert\geq\lvert O^*\cap I_M\rvert$. By induction hypothesis, we know that
\begin{equation*}
\lvert S(I_L)\rvert\geq\frac{\lvert O^*_L\rvert}{\log(n/2)}\geq\frac{\lvert O^*\cap I_L\rvert}{\log n - 1} \mbox{ and } \lvert S(I_R)\rvert\geq\frac{\lvert O^*_R\rvert}{\log(n/2)}\geq\frac{\lvert O^*\cap I_R\rvert}{\log n - 1}.
\end{equation*}
Since $S(I_L)\cap S(I_R)=\emptyset$, we have
\begin{align*}
\lvert S\rvert=\max\{\lvert S(I_L)\cup S(I_R)\rvert, \lvert S(I_M)\rvert\} & \geq\max\{\frac{\lvert O^*\cap I_L\rvert + \lvert O^*\cap I_R\rvert}{\log n-1}, \lvert O^*\cap I_M\rvert\}\\
& \geq\max\{\frac{\lvert O^*\rvert - \lvert O^*\cap I_M\rvert}{\log n-1}, \lvert O^*\cap I_M\rvert\}.
\end{align*}
If $\lvert O^*\cap I_M\rvert\geq\lvert O^*\rvert/\log n$, then we are done. Otherwise,
\begin{align*}
\frac{\lvert O^*\rvert - \lvert O^*\cap I_M\rvert}{\log n-1}\geq\frac{\lvert O^*\rvert - \lvert O^*\rvert/\log n}{\log n-1}=\frac{\lvert O^*\rvert}{\log n},
\end{align*}
which proves the induction step.
\end{proof}

\paragraph{Original problem.} To approximate the original problem in which \pivi consists of all types of paths (i.e., \{\uchap, \dchap, \urast, \drast, $|$, $-$\}), we run the above algorithm four times once for each of the four types of paths (recall that we defined no-bend paths as \dchap-type paths), and then return the largest solution as the final answer. By Lemma~\ref{lem:approxMISonVPGfactor}, our final answer is at least $\frac{1}{4}(\log n)$ times an optimal solution for the original problem, where $\lvert$\pivi$\rvert=n$. This gives us the main result of this section.
\begin{theorem}
\label{thm:approxMISonVPG}
There exists an $O(n^3)$-time $O(\log n)$-approximation algorithm for the \mis problem on any \bvpg graph with $n$ vertices.
\end{theorem}

\section{Dominating Set}
\label{sec:dominatingSet}
In this section, we consider the \mds problem on \emph{one-string} \bvpg and \bepg graphs. We first give our $O(1)$-approximation algorithm for this problem on one-string \bvpg graphs in Section~\ref{subsec:MDSonBVPG}, and then will present our results for this problem on \bepg graphs in Section~\ref{subsec:MDSonBEPG}.

\subsection{One-string \bvpg Graphs}
\label{subsec:MDSonBVPG}
Recall that Asinowski et al.~\cite{AsinowskiCGLLS12} proved that every circle graph is a one-string \bvpg graph. Since \mds is known to be \apx on circle graphs~\cite{DamianP06}, the \mds problem becomes \apx also on one-string \bvpg graphs. In the following, we give an $O(1)$-approximation algorithm for \mds on one-string \bvpg graphs. To our knowledge, this is the first $O(1)$-approximation algorithm for this problem on such graphs. For the rest of this section, let $G$ be a one-string \bvpg graph with $n$ vertices, and let $\langle$\pivi$,\mathcal{G}\rangle$ denote its string representation.

Our algorithm is based on first formulating the problem on $G$ as a hitting set problem and then computing a small $\varepsilon$-net for the corresponding instance of the hitting set problem. Having such an $\varepsilon$-net along with the technique of Br\"{o}nnimann and Goodrich~\cite{BronnimannG95} gives us an $O(1)$-approximation algorithm for \mds on $G$. We first formulate \mds on $G$ as a hitting set problem and then describe our $\varepsilon$-net for it.

\begin{wrapfigure}{r}{0.35\textwidth}
\centering
\includegraphics[scale=0.85]{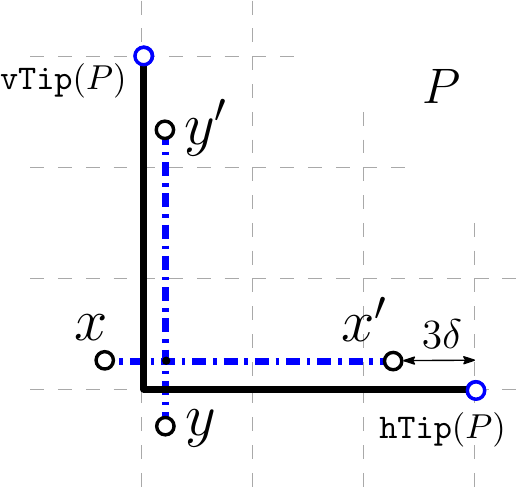}
\label{fig:aCross}
\end{wrapfigure}
\paragraph{Hitting set.} A \emph{set system} is a pair $\mathcal{R}=(\mathcal{U}, \mathcal{S})$, where $\mathcal{U}$ is a ground set of elements and $\mathcal{S}$ is a collection of subsets of $\mathcal{U}$. A \emph{hitting set} for the set system $(\mathcal{U}, \mathcal{S})$ is a subset $M$ of $\mathcal{U}$ such that $M\cap S\neq \emptyset$ for all $S\in\mathcal{S}$; we call each element of $\mathcal{U}$ a \emph{hitting element}.
For the \mds problem on $G$, we construct the set system $(\mathcal{U}, \mathcal{S})$ as follows. Let $P$ be any path in \pivi and assume w.l.o.g. that it is a \dchap-type path. We associate a \emph{cross} $c:=(\ell_H, \ell_V)$ to $P$ in which $\ell_H$ and $\ell_V$ denote its horizontal and vertical segments, respectively; we call $\ell_H$ and $\ell_V$ the \emph{supporting segments} of $c$. The left endpoint of $\ell_H$ is $x:=(x(\cor(P))-\delta, y(\cor(P)))$ and its right endpoint is $x':=((x(\htip(P))-3\delta, y(\cor(P)))$ in which $\delta:=s/4$, where $s$ is the length of a grid edge. See the figure on right for an illustration. Similarly, the bottom endpoint of $\ell_V$ is $y:=(x(\cor(P), y(\cor(P))-\delta)$ and its top endpoint is $y':=((x(\cor(P)), y(\vtip(P))-3\delta)$. We denote the cross of $P$ by $\cro(P)$. The following observation is immediate by the construction of a cross.
\begin{observation}
\label{obs:crossAndPath}
Let $P_u, P_v$ be two paths in \pivi. Then, $P_u$ and $P_v$ intersect each other at a grid node if and only if $\cro(P_u)$ and $\cro(P_v)$ intersect each other.
\end{observation}

To see the hitting elements of the set system, for each path $P\in$\pivi, we add into $\mathcal{U}$ both of the supporting segments of $\cro(P)$. We say that a hitting element $e\in \mathcal{U}$ \emph{hits} a cross $c$ if $e$ intersects one of the supporting segments of $c$; we assume that the hitting elements corresponding to supporting segments of $c$ also hit $c$. To compute $\mathcal{S}$, consider the set of elements of $\mathcal{U}$ hitting a cross $c$ and let $\mathcal{S}$ be the collection of these sets; notice that there is a set in $\mathcal{S}$ for each cross $c$. This forms the set system $(\mathcal{U}, \mathcal{S})$ corresponding to \mds on $G$. For a set $S\in\mathcal{S}$, we say that an element $e\in\mathcal{U}$ \emph{intersects} $S$ if $e\in S$.
\begin{lemma}
\label{lem:ifMDSthenMHS}
If there is a dominating set of size $k$ on the one-string \bvpg graph $G$, then there is a hitting set of size at most $2k$ on the set system $(\mathcal{U}, \mathcal{S})$.
\end{lemma}
\begin{proof}
\label{prf:lem:ifMDSthenMHS}
Let $M$ be a dominating set of size $k$ for $G$. For each path $P\in M$, add the supporting segments of $\cro(P)$ into $M'$. Clearly, $\lvert M'\rvert=2k$. Now, consider a set $S\in\mathcal{S}$ and let $P$ and $c$ denote its corresponding path and cross, respectively. Since $M$ is a dominating set, either $P\in M$ or $P'\in M$ for some path $P'$ intersecting $P$. If $P\in M$, then clearly $S$ is intersected by some segment in $M'$. If $P'\in M$, then $c$ intersects $\cro(P')$ by Observation~\ref{obs:crossAndPath} and so at least one of the supporting segments of $\cro(P')$ (that is in $M'$) intersects $S$.
\end{proof}

\begin{lemma}
\label{lem:ifMHSthenMDS}
If there is a hitting set of size $k$ on the set system $(\mathcal{U}, \mathcal{S})$, then there is a dominating set of size at most $k$ on $G$.
\end{lemma}
\begin{proof}
\label{prf:lem:ifMHSthenMDS}
Suppose that $M'$ is a hitting set of size $k$ for the set system $(\mathcal{U}, \mathcal{S})$. For each hitting element $e\in M'$: add $P$ to $M$, where $P$ is the path for which $e$ is a supporting segment of $\cro(P)$. Clearly, $\lvert M\rvert\leq k$. To see why $M$ is a dominating set for $G$, consider any path $P$ and let $S\in\mathcal{S}$ denote its corresponding set. Since $M'$ is a hitting set, there must be a hitting element $e'\in M'$ that intersects $S$. Let $P'$ be the path for which $e'$ is a supporting segment of $\cor(P')$; notice that $P'\in M$. If $P'=P$, then we are done. Otherwise, $e'$ must hit the cross $c$ corresponding to $P$ since $c\in S$. This means that the cross corresponding to $P'$ intersects $c$ and so $P$ is intersected by $P$ by Observation~\ref{obs:crossAndPath}.
\end{proof}

\paragraph{Approximation using $\varepsilon$-nets.} An {\em $\varepsilon$-net} for a set system $\mathcal{R}=(\mathcal{U}, \mathcal{S})$ is a subset $N$ of $\mathcal{U}$ such that every set $S$ in $\mathcal{S}$ with size at least $\varepsilon\cdot\lvert\mathcal{U}\rvert$ has a non-empty intersection with $N$. Br\"{o}nnimann and Goodrich~\cite{BronnimannG95} introduced the following iterative-doubling approach for turning algorithms for finding $\varepsilon$-nets into approximation algorithms for hitting sets of minimum size. For a given set system $\mathcal{R}=(\mathcal{U},\mathcal{S})$ and any $r>0$, a \emph{net finder} is a polynomial-time algorithm that computes an $(1/r)$-net of $\mathcal{R}$ whose size is at most $s(r)$ for some computable function $s$. Moreover, given a subset $H\subset \mathcal{U}$, a {\em verifier} is a polynomial-time algorithm that outputs whether the set $H$ is a hitting set; if it is not, then the verifier returns a non-empty set $S\in \mathcal{S}$ for which $H$ does not hit $S$.
\begin{theorem}\cite{BronnimannG95}
\label{thm:epsNet}
Let $\mathcal{R}$ be a set system that admits both a polynomial-time net finder and a polynomial-time verifier. Then, there exists a polynomial-time algorithm that computes a hitting set of size at most $s(4 \cdot \textsc{OPT})$, where \textsc{OPT} is the size of a minimum hitting set, and $s(r)$ is the size of the $(1/r)$-net.
\end{theorem}

Clearly, the hitting set problem corresponding to \mds on $G$ has a polynomial-time verifier. In the following, we compute an $\varepsilon$-net finder of size $O(1/\varepsilon)$ for the set system $(\mathcal{U}, \mathcal{S})$. Then, we prove that this result combined by Lemmas~\ref{lem:ifMDSthenMHS}---\ref{lem:ifMHSthenMDS} and Theorem~\ref{thm:epsNet}, gives an $O(1)$-approximation algorithm for \mds on $G$.

Unfortunately, we were unable to compute an $\varepsilon$-net of size $O(1/\varepsilon)$ directly for $(\mathcal{U},\mathcal{S})$. Instead, we first compute such an $\varepsilon$-net for a ``one-dimensional'' variant of our hitting set problem and then show that re-using such an $\varepsilon$-net twice would result in the desired $\varepsilon$-net. Informally speaking, the idea is that although the \mds problem on $G$ is not ``decomposable'' into horizontal and vertical instances, the $\varepsilon$-net corresponding to our set system $(\mathcal{U},\mathcal{S})$ for this problem is. In the following, we first formulate a one-dimensional variant of our hitting set problem and then show that it admits a small-size $\varepsilon$-net.

Recall the set system $(\mathcal{U}, \mathcal{S})$ corresponding to \mds on $G$ that we constructed just before Lemma~\ref{lem:ifMDSthenMHS}. Let $\mathcal{U}_H$ denote the set of only-horizontal hitting elements of $\mathcal{U}$, and consider the set system $(\mathcal{U}_H, \mathcal{S})$. Representing each cross by its vertical supporting segment only, the minimum hitting set problem on $(\mathcal{U}_H, \mathcal{S})$ reduces to the following problem: given a set of horizontal line segments $\mathcal{H}$ and a set of vertical line segments $\mathcal{V}$, find a minimum-cardinality set $S\subseteq\mathcal{H}$ such that every line segment in $\mathcal{V}$ is intersected by $S$. This problem is known as the \emph{orthogonal segment covering} problem and is known to be \textsc{NP}-complete~\cite{KatzMN05}.
\begin{lemma}
\label{lem:oneDmdsVhitting}
The minimum hitting set problem on $(\mathcal{U}_H, \mathcal{S})$ reduces to the orthogonal segment covering problem.
\end{lemma}
\begin{proof}
\label{prf:lem:oneDmdsVhitting}
($\Rightarrow$) Let $M$ be a feasible solution for the minimum hitting set problem on $(\mathcal{U}_H, \mathcal{S})$. If a hitting element $e$ hits a cross $c$, then it intersects one of its supporting segments $s$. Let $P$ denote the path in \pivi corresponding to cross $c$. Since $e$ is horizontal, $s$ cannot be horizontal because if $e$ intersects the horizontal supporting segment of $P$ (i.e., $s$), then $\vpart(P)$ and $\hpart(P')$ must intersect each other in more than one point, where $P'$ is the path in \pivi for which $e$ is a supporting segment. This contradicts the fact that $G$ is a one-string \bepg graph. Therefore, $s$ is vertical and so for any feasible solution $M$, we have a feasible solution for the orthogonal segment covering problem with the same size.

($\Leftarrow$) Clearly, a feasible solution to the minimum segment covering problem is also a feasible solution to the minimum hitting set problem with the same size.
\end{proof}

For the orthogonal segment covering problem, Biedl et al.~\cite{BiedlCLMMV16} showed (using the linear union-complexity ideas of Clarkson and Varadarajan~\cite{ClarksonV07}) that there exists an $\varepsilon$-net finder of size $s(1/\varepsilon)\in O(1/\varepsilon)$. Therefore, by Lemma~\ref{lem:oneDmdsVhitting}, we have the following result.
\begin{lemma}
\label{lem:epsilonUH}
There exists an $\varepsilon$-net of size $O(1/\varepsilon)$ for the minimum hitting set problem on $(\mathcal{U}_H, \mathcal{S})$.
\end{lemma}

\begin{lemma}
\label{lem:approxMHSonU}
There exists a polynomial-time $O(1)$-approximation algorithm for the minimum hitting set problem on $(\mathcal{U},\mathcal{S})$.
\end{lemma}
\begin{proof}
\label{prf:lem:approxMHSonU}
To prove the lemma, it suffices by Theorem~\ref{thm:epsNet} to compute an $\varepsilon$-net of size $O(1/\varepsilon)$ for $(\mathcal{U},\mathcal{S})$. Define the set system $(\mathcal{U}_V, \mathcal{S})$ similar to $(\mathcal{U}_H, \mathcal{S})$; that is, $\mathcal{U}_V$ is the set of only-vertical hitting elements of $\mathcal{U}$. By Lemma~\ref{lem:epsilonUH}, let $N_H$ (resp., $N_V$) be an $(\varepsilon/2)$-net of size $O(1/\varepsilon)$ for the corresponding one-dimensional variant of the hitting set problem on $(\mathcal{U}_H, \mathcal{S})$ (resp., $(\mathcal{U}_V, \mathcal{S})$). Let $N:=N_H\cup N_V$. Clearly, $N$ has size $O(1/\varepsilon)$. We now prove that $N$ is an $\varepsilon$-net.

Let $S\in\mathcal{S}$ such that $\lvert S\rvert\geq\varepsilon\cdot\lvert\mathcal{U}\rvert$. We need to show that $N\cap S\neq\emptyset$. Notice that having $\lvert S\rvert\geq\varepsilon\cdot\lvert\mathcal{U}\rvert$ means that there exists a cross $c$ that is hit by $\varepsilon\cdot\lvert\mathcal{U}\rvert$ hitting elements. Assume w.l.o.g. that at least half of these hitting elements are horizontal. Then, the vertical supporting segment $s$ of $c$ intersects at least $\varepsilon\cdot\lvert\mathcal{U}/2\rvert$ horizontal hitting elements of $\mathcal{U}$. By definition of an $(\varepsilon/2)$-net, there is a hitting element $e\in N_H$ that intersects $s$. Therefore, $e$ hits cross $c$ and so $e\in N\cap S$.
\end{proof}

Putting everything together, we can prove the main result of this section.
\begin{theorem}
\label{thm:approxMDSonVPG}
There exists a polynomial-time $O(1)$-approximation algorithm for \mds on any one-string \bvpg graph.
\end{theorem}
\begin{proof}
Let $G$ be any one-string \bvpg graph with its string representation $\langle$\pivi$,\mathcal{G}\rangle$. We show how to obtain a constant-factor approximation solution for \mds on $G$ in polynomial time and so proving the theorem.

First, compute the set system $(\mathcal{U},\mathcal{S})$ corresponding to \pivi as described above. Let $\optds$ (resp., $\opths$) denote an optimal solution for \mds on $G$ (resp., for the minimum hitting set problem on $(\mathcal{U},\mathcal{S})$). By Lemma~\ref{lem:ifMDSthenMHS}, we know that $\lvert\opths\rvert\leq 2\lvert\optds\rvert$. By Lemma~\ref{lem:approxMHSonU}, let $S_{HS}$ be a constant-factor approximation to the minimum hitting set problem on $(\mathcal{U},\mathcal{S})$; that is, $\lvert S_{HS}\rvert\leq c\cdot\lvert \opths\rvert$ for some constant $c>0$. Apply Lemma~\ref{lem:ifMHSthenMDS} to $S_{HS}$ and let $S_{DS}$ be a feasible solution for \mds on $G$; notice that $\lvert S_{DS}\rvert\leq\lvert S_{HS}\rvert$. The set system $(\mathcal{U},\mathcal{S})$, $S_{HS}$ and $S_{DS}$ can each be computed in polynomial time. Therefore,
\begin{equation*}
\frac{\lvert S_{DS}\rvert}{\lvert\optds\rvert}\leq\frac{\lvert S_{HS}\rvert}{\lvert\optds\rvert}\leq \frac{2\lvert S_{HS}\rvert}{\lvert\opths\rvert}\leq 2c.
\end{equation*}

That is, $S_{DS}$ is an $O(1)$-approximation solution for \mds on $G$, which can be computed in polynomial time.
\end{proof}

\subsection{\bepg Graphs}
\label{subsec:MDSonBEPG}
In this section, we consider the \mds problem on \bepg graphs. We first show that the problem is \apx on \bepg graphs, even if only two types of paths are allowed in the graph; hence, ruling out the possibility of a PTAS for this problem unless \textsc{P=NP}. Then, we give $c$-approximation algorithms for this problem on two subclasses of \bepg graphs, for $c\in\{2, 3\}$.

\paragraph{\textsc{APX}-hardness.} To show the \textsc{APX}-hardness, we give an \textsc{L}-reduction from the \mvc problem on graphs with maximum-degree three to the \mds on \bepg graphs; \mvc is known to be \apx on graphs with maximum-degree three~\cite{AlimontiK00}. Before describing the reduction, let us give a formal definition of \textsc{L}-reduction~\cite{PapadimitriouY91} as a reminder. Let $\Pi$ and $\Pi'$ be two optimization problems with the cost functions $c_\Pi(.)$ and $c_{\Pi'}(.)$, respectively. We say that \emph{$\Pi$ $L$-reduces to $\Pi'$} if there are two polynomial-time computable functions $f$ and $g$ such that the followings hold. \begin{inparaenum} \item For any instance $x$ of $\Pi$, $f(x)$ is an instance of $\Pi'$. \item If $y$ is a solution to $f(x)$, then $g(y)$ is a solution for $x$. \item There exists a constant $\alpha>0$ such that $OPT_{\Pi'}(f(x))\leq\alpha OPT_\Pi(x)$, where $OPT_Y(x)$ denotes the cost of an optimal solution for problem $Y$ on its instance $x$. \item There exists a constant $\beta>0$ such that for every solution $y$ for $f(x)$,
\[
|OPT_\Pi(x)-c_\Pi(g(y))|\leq \beta |OPT_{\Pi'}(f(x))-c_{\Pi}(y)|,
\]
where $|x|$ denotes the absolute value of $x$. \end{inparaenum}

\begin{lemma}
\label{lem:mdsApx}
\mvc on graphs with maximum-degree three is \textsc{L}-reducible to \mds on \bepg graphs.
\end{lemma}
\begin{proof}
\label{prf:lem:mdsApx}
Consider an arbitrary instance $I$ of \mvc on graphs of maximum-degree three; let $G=(V, E)$ be the graph corresponding to $I$ and let $k$ be size of the smallest vertex cover in $G$. First, let $u_1, \dots, u_n$ be an arbitrary ordering of the vertices of $G$, where $n=\lvert V\rvert$. In the following, we give a computable function $f$ that takes $I$ as input and outputs an instance $f(I)$ of \mds in polynomial time, where $f(I)$ consists of a \bepg graph such that its paths are of either \uchap-type or \drast-type only.

We first describe the vertex gadgets. For each vertex $u_i$, $1\leq i\leq n$, construct a horizontal \uchap-type $\Gamma^h_i$ and a vertical \uchap-type $\Gamma^v_i$, and connect them as shown in Figure~\ref{fig:mdsApx}. We call the big \uchap-type path used in the connection of $\Gamma^h_i$ and $\Gamma^v_i$ the \emph{big connector} $C_i$ of $i$, and the two small (blue, dashed) \uchap-type paths the \emph{small connectors} of $i$. For each edge $(u_i, u_j)\in E$, where $i<j$, we add two small paths, one of \uchap-type and one of \drast-type, at the intersection point of $\Gamma^v_i$ and $\Gamma^h_j$ such that each of them becomes adjacent to both $\Gamma^v_i$ and $\Gamma^h_j$; see the two (red, dash-dotted) \uchap-type and \drast-type paths at the intersection of $\Gamma^v_1$ and $\Gamma^h_2$ in Figure~\ref{fig:mdsApx}. We denote this pair of paths by $E_{i,j}$; notice that the paths of $E_{i,j}$ are not adjacent to each other (they only share a grid node). This gives the instance $f(I)$ of \mds on \bepg graphs; let $G'$ be the corresponding \bepg graph. Notice that $f$ is a polynomial-time computable function. In the following, we denote an optimal solution for the instance $X$ of a problem by $s^*(X)$. We now prove that all the four conditions of $L$-reduction hold.

First, let $M$ be a vertex cover of $G$ of size $k$. Denote by $\Gamma^h[M] := \{\Gamma^h_i | u_i\in M\}$ the set of horizontal paths induced by $M$ and define $\Gamma^v[M]$ analogously. Moreover, let $C[M] := \{C_i | u_i \notin M\}$ be the set of big connectors whose corresponding vertex is not in $M$. We show that $D := \Gamma^h[M]\cup \Gamma^v[M]\cup C[M]$ is a dominating set of $G'$. Let $\gamma$ be a path. If $\gamma$ is any of the paths in $E_{i,j}$ for some $i,j$, then $u_i\in M$ or $u_j\in M$ because $M$ is a vertex cover; assume w.l.o.g. that $u_i\in M$. Then, $\Gamma^h_i, \Gamma^v_i\in D$ and so $\gamma$ must be dominated. If $\gamma$ is a connector of $i$ for some $i$ (either big or small), then there are two cases: if $u_i\in M$, then $\Gamma^h_i, \Gamma^v_i\in D$ and so the connector is dominated. If $u_i\notin M$, then $C_i\in C[M]$ and $C[M]\subseteq D$; hence, the connector is again dominated. Finally, suppose that $\gamma$ is $\Gamma^h_i$ (resp., $\Gamma^v_i$) for some $i$. If $u_i\in M$, then $\Gamma^h_i, \Gamma^v_i\in D$ and so $\Gamma^h_i$ (resp., $\Gamma^v_i$) is dominated. If $u_i\notin M$, then $C_i\in C[M]$ and $C[M]\subseteq D$; hence, $\Gamma^h_i$ (resp., $\Gamma^v_i$) is again dominated. This shows that $D$ is a dominating set for $G'$.

\begin{figure}[t]
\centering
\includegraphics[width=0.50\textwidth]{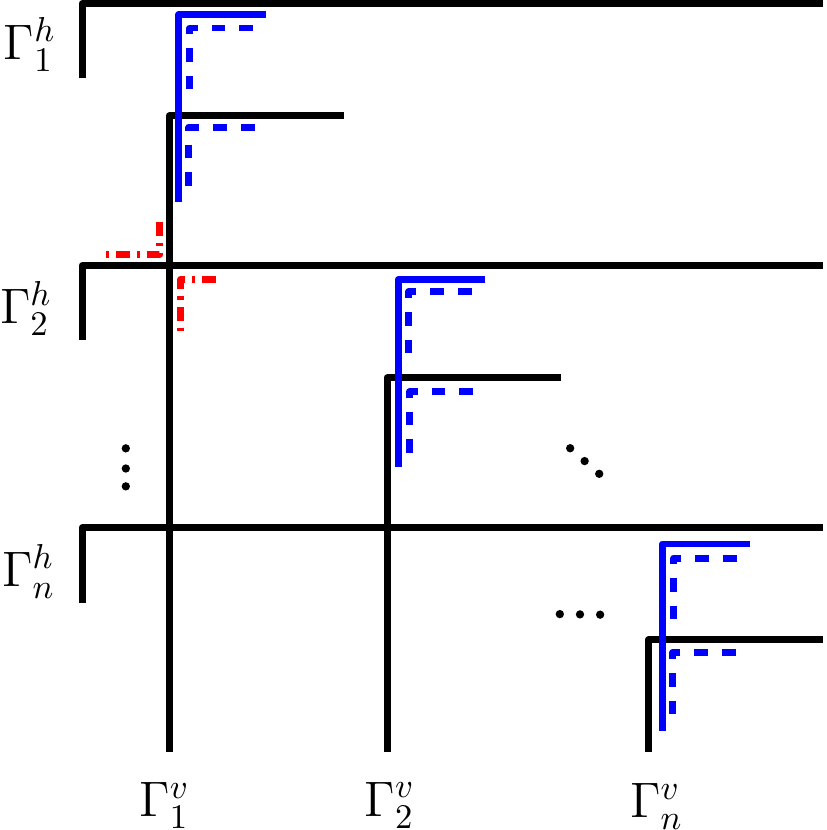}
\caption{An illustration in support of the construction in the proof of Lemma~\ref{lem:mdsApx}.}
\label{fig:mdsApx}
\end{figure}

Second, let $D$ be an arbitrary dominating set on $G'$. First, notice that we can construct a dominating set $D'$ for $G'$ such that $\lvert D'\rvert\leq\lvert D\rvert$ and $D'$ consists of only $\Gamma^h_i$ and $\Gamma^v_i$ for some $i$ or a big connector. This is because \begin{inparaenum}[(i)] \item any path dominated by a small connector is also dominated by some big connector, and \item\label{inl:domiantingEij} any path dominated by a path from $E_{i,j}$ is also dominated by $\Gamma^v_i$ or $\Gamma^h_j$. \end{inparaenum} For~\eqref{inl:domiantingEij}, in particular, if exactly one of the paths in $E_{i,j}$ is in $D$, then we can replace it with one of $\Gamma^v_i$ or $\Gamma^h_j$ arbitrarily. Moreover, if both paths in $E_{i,j}$ are in $D$, then we can replace both of them with $\Gamma^v_i$ and $\Gamma^h_j$. As such, $\lvert D'\rvert\leq\lvert D\rvert$ and $D'$ is a feasible dominating set for $G'$. Now, define $\Gamma_{\both}[D'] = \{\Gamma^h_i,\Gamma^v_i | \Gamma^h_i,\Gamma^v_i \in D'\}$; i.e., the paths of a vertex $u_i$, where \emph{both} its horizontal and its vertical copies appear in $D'$. Also, define $\Gamma_{\one}[D']$ to be the remaining paths of type $\Gamma^h_i$ and $\Gamma^v_i$; i.e., those of $u_i$, where \emph{only} one of its copies appears in $D'$. Finally, let $C[D']$ be the set of big connectors in $D'$. We denote $\Gamma_{\both}[D']\cup\Gamma_{\one}[D']$ by $\Gamma[D']$. Now, let $M := \{u_i | \Gamma^h_i\in D' \mbox{ or } \Gamma^v_i\in D'\}$. Since all $E_{i,j}$ are dominated by $\Gamma[D']$, $M$ is a vertex cover.

Third, observe that $\lvert\Gamma^h[M]\rvert=\lvert\Gamma^v[M]\rvert=\lvert M\rvert=k$ and also $\lvert C[M]\rvert=n-k$. Given that $G$ has degree three, $k\geq n/4$ and so $\lvert s^*(f(I))\rvert\leq n-k+k+k\leq 5k\leq 5\lvert s^*(I)\rvert$.

Finally, to dominate all connectors of $i$, we must have $C_i\in D'$ or $\Gamma^h_i, \Gamma^v_i\in D'$; this indeed holds for all $i$. Thus, $\lvert C[D']\rvert+\lvert \Gamma_{\both}[D']\rvert/2\geq n$. Moreover, $\lvert\Gamma_{\one}[D']\rvert+\lvert\Gamma_{\both}[D']\rvert/2\geq k$ since $M$ is a vertex cover of $G$. Therefore, $\lvert D'\rvert\geq\lvert\Gamma_{\both}[D']\rvert+\lvert\Gamma_{\one}[D']\rvert+\lvert C[D']\rvert\geq\lvert\Gamma_{\one}[D']\rvert+\lvert\Gamma_{\both}[D']\rvert/2+n\geq k+n$. By this and our earlier inequality $\lvert s^*(f(I))\rvert\leq n-k+k+k$, we have $\lvert s^*(f(I))\rvert=n+k$. Now, suppose that $\lvert D\rvert=\lvert s^*(f(I))\rvert+c$ for some $c\geq 0$. Then,
\begin{align*}
& \lvert D\rvert-\lvert s^*(f(I))\rvert=c\\
& \Rightarrow \lvert D\rvert-(n+k)=c\\
& \Rightarrow \lvert D'\rvert-(n+k)\leq c\\
& \Rightarrow \lvert\Gamma_{\one}[D']\rvert+\lvert\Gamma_{\both}[D']\rvert/2+n-(n+k)\leq c\\
& \Rightarrow \lvert\Gamma_{\one}[D']\rvert+\lvert\Gamma_{\both}[D']\rvert/2-k\leq c\\
& \Rightarrow \lvert M\rvert-\lvert s^*(I)\rvert\leq c.
\end{align*}
That is, $\lvert M\rvert-\lvert s^*(I)\rvert\leq\lvert D\rvert-\lvert s^*(f(I))\rvert$. This concludes our \textsc{L}-reduction from \mvc on graphs of maximum-degree three to \mds on \bepg graphs with $\alpha=5$ and $\beta=1$.
\end{proof}

Our reduction reveals that every path in the constructed \bepg graph $G'$ is a \uchap-type or a \drast-type path. Therefore, by Lemma~\ref{lem:mdsApx}, we have the following theorem.
\begin{theorem}
\label{thm:mdsApx}
The \mds problem is \apx on \bepg graphs, even if all the paths in the graph are \texttt{x}-type paths, for some $\texttt{x}\in\{$\uchap, \drast$\}$. Thus, there exists no PTAS for this problem on \bepg graphs unless \textsc{P=NP}.
\end{theorem}

\paragraph{Approximation algorithms.} In this section, we give constant-factor approximation algorithms for the \mds problem on two subclasses of \bepg graphs. Let us first define these subclasses.

First, we consider a subclass of \bepg graphs in which every path of each \bepg graph intersects two axis-parallel lines that are normal to each other. Recall that this is the variant for which Lahiri et al.~\cite{LahiriMS15} gave an exact solution for the \mis problem when the input graph is a \bvpg graph: they showed that the induced graph is a co-comparability graph and so solved the \mis problem exactly. However, the graph induced by this variant when considering \bepg graphs is not necessarily a co-comparability graph; this is mainly because two paths intersecting in only one point in a \bepg graph are not adjacent. Here, we give a 2-approximation algorithm for this problem on such \bepg graphs. We call this subclass, the class of \textsc{Double-Crossing} \bepg graphs. Next, we consider a less-restricted subclass of \bepg graphs in which every path of each \bepg graph intersects only a vertical line. We show that the same algorithm is a 3-approximation algorithm for the problem on this subclass of \bepg graphs, albeit considering a ``non-containment'' assumption. We call this subclass, the class of \textsc{Vertical-Crossing} \bepg graphs.

Before describing the algorithms, let us define an ordering $\prec$ on the paths in \piei as follows. The paths appear in the ordering by the $y$-coordinate of their corners from bottom to top and then from left to right whenever they have the same $y$-coordinate; that is, $P\prec P'$ for two paths in \piei, if and only if $y(\cor(P))<y(\cor(P'))$ or $y(\cor(P))=y(\cor(P'))$ but $x(\cor(P))<x(\cor(P'))$. Notice that this ordering is well defined by out weak general-position assumption. For the rest of this section, we assume that every path in \piei is a \dchap-type path.

\begin{wrapfigure}{r}{0.35\textwidth}
\centering
\includegraphics[scale=0.8]{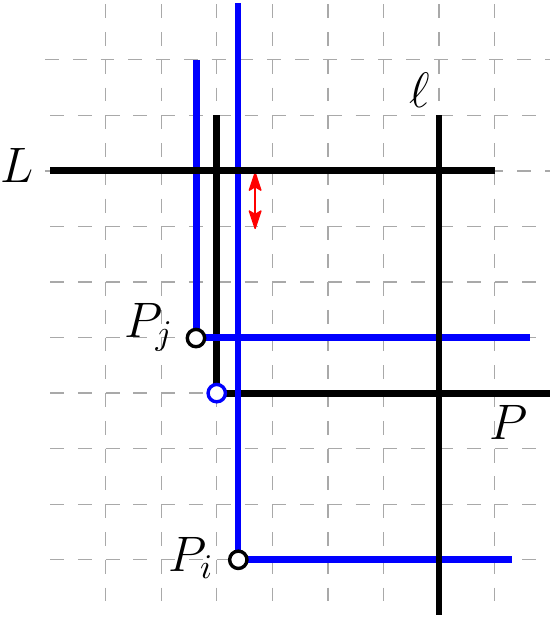}
\label{fig:dcMDS}
\end{wrapfigure}

\paragraph{Double-crossing \bepg graphs.}
Recall that in a \textsc{Double-Crossing} \bepg graph, we are given a \bepg graph $G$, a horizontal line $L$ and a vertical line $\ell$ both on the grid $\mathcal{G}$ such that $L$ and $\ell$ intersect each other and $P$ intersects both $L$ and $\ell$ for all $P\in$\piei (hence, $\cor(P)$ lies in the lower-left quadrant defined by $L$ and $\ell$). Our 2-approximation algorithm for the \mds problem is as follows; let $S$ be an initially-empty set. For each path $P$ in the increasing order $\prec$: add $P$ into $S$ and set \piei$:=$\piei$\setminus N[P]$. See Algorithm~\ref{alg:approxLineMDS}. Clearly, the algorithm terminates in time polynomial in $\lvert$\piei$\rvert$, and $S$ is a feasible solution for the problem. To see the approximation factor, let $OPT$ be an optimal solution for the \mds problem on $G$; notice that by deleting the paths in $S\cap OPT$ we can assume that $S\cap OPT=\emptyset$. This means that every path in $S$ must be adjacent to at least one path in $OPT$.

\begin{lemma}
\label{lem:2crossingMDS}
Any path in $OPT$ is adjacent to at most two distinct paths in $S$.
\end{lemma}
\begin{proof}
Suppose for a contradiction that there exists a path $P\in OPT$ that is adjacent to three distinct paths $P_1, P_2$ and $P_3$ of $S$. W.l.o.g., assume that $P_i, P_j\in \vnei(P)$ for some $i<j\in\{1,2,3\}$. This means that $x(\cor(P))=x(\cor(P_i))=x(\cor(P_j))$. Since $i<j$, we have $y(\cor(P_i))<y(\cor(P_j))$. Moreover, since all the paths in \piei intersect the horizontal line $L$, the three paths $P, P_i$ and $P_j$ must all intersect $L$ at the same point and so they share the topmost vertical grid edge below $L$ on which $\vpart(P)$ lies; see the figure on right. Thus, $P_i$ and $P_j$ are adjacent in $G$. Since $y(\cor(P_i))<y(\cor(P_j))$, we have $P_i\prec P_j$ and so $P_j$ is removed from \piei when the algorithm adds $P_i$ to $S$. So, $P_j\notin S$ --- a contradiction.
\end{proof}

Since every path in $S$ must be adjacent to at least one path in $OPT$ and any path in $OPT$ can be adjacent to at most two distinct paths in $S$ by Lemma~\ref{lem:2crossingMDS}, we have $\lvert S\rvert\leq 2\lvert OPT\rvert$. Therefore, we have the following theorem.
\begin{theorem}
\label{thm:2crossingMDS}
There exists a polynomial-time 2-approximation algorithm for the \mds problem on \textsc{Double-Crossing} \bepg graphs.
\end{theorem}

\begin{algorithm}[t]
\caption{\textsc{ApproximateLineMDS}($G$, $L$, $\ell$)}
\label{alg:approxLineMDS}
\begin{algorithmic}[1]
\State $S\gets \emptyset$;
\For{each path $P\in$\piei in increasing order $\prec$}
	\State $S\gets S\cup P$;
	\State \piei$\gets$\piei$\setminus N[P]$;
\EndFor
\State\Return $S$;
\end{algorithmic}
\end{algorithm}

\paragraph{Vertical-crossing \bepg graphs.}
Here, we are given a \bepg graph $G$ and a vertical line $\ell$ on the grid $\mathcal{G}$ such that $\hpart(P)$ intersects $\ell$ for all $P\in$\piei. Moreover, we make a \emph{non-containment assumption} in the sense that the vertical segment of no path is entirely contained in that of any other path in \piei; that is, for every two paths $P, P'\in$\piei such that $P\in\vnei(P')$, neither $\vpart(P)\subseteq\vpart(P')$ nor $\vpart(P')\subseteq\vpart(P)$. We prove that Algorithm~\ref{alg:approxLineMDS} is a 3-approximation algorithm for the \mds problem on \textsc{Vertical-crossing} \bepg graphs.
\begin{theorem}
\label{thm:approxGPrim}
Algorithm~\ref{alg:approxLineMDS} is a 3-approximation algorithm for the \mds problem on \textsc{Vertical-crossing} \bepg graphs.
\end{theorem}
\begin{proof}
\label{prf:thm:approxGPrim}
Let $G$ be any \textsc{Vertical-Crossing} \bepg graph. Moreover, let $OPT$ be an optimal solution for the \mds problem on $G$ and let $S$ be the solution returned by Algorithm~\ref{alg:approxLineMDS}. Again, we can assume that $S\cap OPT=\emptyset$. Thus, every path in $S$ must be dominated by at least one path in $OPT$. In the following, we show that any path in $OPT$ can dominate at most three distinct paths in $S$ and so prove that $\lvert S\rvert\leq 3\lvert OPT\rvert$. Let $P$ be any path in $OPT$. We show that $\lvert S\cap \hnei(P)\rvert\leq 1$ and $\lvert S\cap \vnei(P)\rvert\leq 2$.

First, suppose for a contradiction that there are two paths $P_1,P_2\in S\cap\hnei(P)$; assume w.l.o.g. that $P_1\prec P_2$. Notice that $y(\cor(P))=y(\cor(P_1))=y(\cor(P_2))$. Thus, since all the three paths $P, P_1$ and $P_2$ intersect the vertical line $\ell$, we conclude that they all intersect $\ell$ at the same point. Therefore, they share the rightmost horizontal grid edge to the left of $\ell$ on which $\hpart(P)$ lies. This means that $P_1$ and $P_2$ are adjacent in $G$. Since $P_1\prec P_2$, Algorithm~\ref{alg:approxLineMDS} removes $P_2$ from \piei when adding $P_1$ into $S$; that is, $P_2\notin S$ --- a contradiction. So, $\lvert S\cap \hnei(P)\rvert\leq 1$.

Now, suppose for a contradiction that there are three paths $P_1, P_2, P_3\in S\cap\vnei(P)$; assume w.l.o.g. that $P_1\prec P_2\prec P_3$. Notice that $x(\cor(P))=x(\cor(P_i))$ for all $1\leq i\leq 3$. Consider $y(\cor(P))$. Then, for at least two paths $P_i, P_j$, where $i<j\in\{1,2,3\}$, we have either $y(\cor(P_i)), y(\cor(P_j))>y(\cor(P))$ or $y(\cor(P_i)), y(\cor(P))<y(\cor(P))$ (equality does not happen due to our weak general-position assumption).
\begin{itemize}
\item If $y(\cor(P_i)), y(\cor(P_j))<y(\cor(P))$, then $\vpart(P_i)$ and $\vpart(P_j)$ both share with $P$ the bottommost vertical grid edge on which $\vpart(P)$ lies, implying that $P_i$ and $P_j$ are adjacent in $G$. Since $i<j$, we have $P_i\prec P_j$ and so Algorithm~\ref{alg:approxLineMDS} removes $P_j$ from $G$ when adding $P_i$ to $S$. So, $P_j\notin S$ --- a contradiction.
\item If $y(\cor(P_i)), y(\cor(P_j))>y(\cor(P))$, then $x(\vtip(P_i))>x(\vtip(P))$ and $x(\vtip(P_j))>x(\vtip(P))$ because otherwise $\vpart(P_i)\subseteq \vpart(P)$ or $\vpart(P_j)\subseteq \vpart(P)$, which is a contradiction to the non-containment assumption of paths in $G$. Since $i<j$, we have $P_i\prec P_j$. Therefore, $P_i$ and $P_j$ share the topmost vertical grid edge on which $\vpart(P)$ lies, meaning that $P_i$ and $P_j$ are adjacent in $G$. So, Algorithm~\ref{alg:approxLineMDS} removes $P_j$ when adding $P_i$ into $S$ --- a contradiction to $P_j\in S$.
\end{itemize}

Therefore, $\lvert S\cap \hnei(P)\rvert\leq 1$ and $\lvert S\cap \vnei(P)\rvert\leq 2$. This completes the proof of the theorem.
\end{proof}

\section{Conclusion}
\label{sec:conclusion}
In this paper, we considered the \mis and \mds problems on \bvpg and \bepg graphs. For \bvpg graphs, we gave an $O(\log n)$-approximation algorithm for the \mis problem, and an $O(1)$-approximation algorithm for the \mds problem when the input graph is one-string. For \bepg graphs, we proved that \mds is \textsc{APX}-hard (even if the graph has only two types of paths), ruling out the existence of a PTAS unless \textsc{P=NP}. We also gave $c$-approximation algorithms for this problem on two subclasses of \bepg graphs, for $c\in\{2, 3\}$. We conclude the paper by the following open problems:
\begin{itemize}
\item Is there an $O(1)$-approximation algorithm for the \mis problem on \bvpg graphs? Also, does the problem admit a PTAS or it is \textsc{APX}-hard?
\item Our $O(1)$-approximation algorithm for \mds on one-string \bvpg graphs relies on the fact that the input graph is one-string (in the proof of Lemma~\ref{lem:oneDmdsVhitting}, in particular); is there an $O(1)$-approximation algorithm for this problem on any \bvpg graph?
\item Although graph $G'$ in our \textsc{APX}-hardness consists of two types of paths, we believe that the problem remains \apx even if only one type of paths are allowed. Is the \mds problem \apx on \bepg graphs, if the graph consists of only one type of paths?
\item Is there an $O(1)$-approximation algorithm for the \mds problem on \bepg graphs?
\item Is the \mds problem \textsc{NP}-hard on \textsc{Vertical-Crossing} \bepg graphs?
\end{itemize}

\bibliographystyle{plain}
\bibliography{ref}

\end{document}